\def\papertitle{Characterization of Incentive Compatible
  Single-parameter Mechanisms Revisited}

\def\paperauthorA{Krzysztof R. Apt}
\def\paperauthorB{Jan Heering}

\documentclass[twoside,12pt]{article}
\usepackage{aux,amssymb,amsmath}

\usepackage{lastpage}
\usepackage{mathptmx}
\usepackage{microtype}
\usepackage[english]{babel}
\newif\ifpdf
\ifx\pdfoutput\relax
\else
   \ifcase\pdfoutput
      \pdffalse
   \else
      \pdftrue
\fi

\ifpdf 
  \usepackage[pdftex, pdftitle={\papertitle},
	pdfauthor={\paperauthorA, \paperauthorB}, 
	colorlinks,
	linkcolor=red,
	citecolor=blue,
	urlcolor=blue,
	bookmarksnumbered, 
	pdfstartview= XYZ	
	]{hyperref}
  \usepackage[pdftex]{graphicx}
  \usepackage{caption}
\else 
  \usepackage[dvips]{graphicx}
  \usepackage[dvips,
	colorlinks,
	linkcolor=red,
	citecolor=blue,
	urlcolor=blue,
	bookmarksnumbered, 
	pdfstartview= XYZ	
   ]{hyperref}
  \usepackage{caption}
\fi

\usepackage{amsmath, amssymb, amsfonts, amsthm, amstext, latexsym, graphicx, multirow, float, rotating, setspace, tabularx, times, url, xpinyin, epsfig}
\usepackage{apacite}
\usepackage{natbib}
\usepackage[english]{babel}
\usepackage[utf8]{inputenc}
\usepackage[titletoc,title]{appendix}
\usepackage[flushmargin]{footmisc}
\usepackage{enumerate}

\newtheorem{theorem}{Theorem}
\newtheorem{defined}[theorem]{Definition}

\newtheorem{exa}{Example}

\newtheorem{lemma}[theorem]{Lemma}
\newtheorem{corollary}[theorem]{Corollary}
\newtheorem{remark}[theorem]{Remark}
\newtheorem{note}[theorem]{Note}

\newtheorem{exe}{Exercise}

\newtheorem{pro}{Problem}

\newcommand{\bfe}[1]{\begin{bfseries}\emph{#1}\end{bfseries}}

\newcommand{\fa}{\mbox{$\forall$}}

\newcommand{\LL}{\mbox{$\ldots$}}

\newcommand{\NI}{\noindent}
\newcommand{\HB}{\hfill{$\Box$}}

\newcommand{\II}{\vspace{2 mm}}

\newcommand{\szkew}[1]{\relax \setbox0=\hbox{\kern -24pt $\displaystyle#1$\kern 0pt }%
\box0}
{\catcode`\@=11 \global\let\ifjusthvtest@=\iffalse}

\newcounter{oldmycaption}

\newcommand{\x}{\textbf{x}}
\newcommand{\y}{\textbf{y}}

\title{\papertitle}

\twoaffiliations{\paperauthorA}{CWI, Amsterdam, The Netherlands\\ MIMUW, University of Warsaw, Poland \\ \texttt{\small apt@cwi.nl}}
{\paperauthorB}{CWI, Amsterdam, The Netherlands (retired) \\ \texttt{\small jan.heering1@gmail.com}\thanks{We thank Guido Sch\"{a}fer for suggesting to analyze the
characterization result discussed in this paper and Marcin
Dziubi\'{n}ski for helpful comments. We are grateful to one of
  the referees for recommending to discuss the uniqueness proofs
  given in \cite{Kri09} and \cite{Bor15}.
}}

\begin{document}


\ifpdf 
  \DeclareGraphicsExtensions{.png,.jpg,.pdf}
\else  
  \DeclareGraphicsExtensions{.eps}
\fi

\maketitle\label{FirstPage}

\begin{abstract}
We reexamine the characterization of incentive compatible single-parameter
mechanisms introduced in \cite{AT01}. We argue that the claimed
uniqueness result, called `Myerson’s Lemma' was not well established. We
provide an elementary proof of uniqueness that unifies the presentation for
two classes of allocation functions used in the literature and show that the
general case is a consequence of a little known result from the theory of real
functions. We also clarify that our proof of uniqueness is more elementary
than the previous one. Finally, by generalizing our characterization result to
more dimensions, we provide alternative proofs of revenue equivalence results
for multiunit auctions and combinatorial auctions.
\end{abstract}

\vspace{0.2in}

\noindent {\it Keywords}: incentive compatibility, single-parameter
mechanisms, Myerson's lemma, auctions, revenue equivalence.

\vspace{.1in}

\noindent {\it JEL Classification Numbers}: D44, D82.


\section{Introduction}


\lettrine[findent=2pt]{\textbf{W}}{ }\ e are concerned here with a characterization result of a specific
class of incentive compatible direct selling mechanisms.  For the sake
of this article such a \bfe{mechanism} consists of an \bfe{allocation
  rule} that assigns some good or goods to the participants, called
agents, and a \bfe{payment rule} that determines how much each agent
needs to pay.  Assuming that each agent has a \bfe{private valuation}
of the good or goods, these decisions are taken in response to a
vector of \bfe{bids} made by the agents. These bids may differ from
agents’ true valuations.  Recall that a mechanism is (\bfe{dominant
  strategy}) \bfe{incentive compatible} (alternatively called
\bfe{truthful}), if no agent is better off when providing false
information regardless of reports of the other agents
or, more precisely, when submitting a bid different from his/her
valuation regardless of reports of the other agents.

Given a class of mechanisms one of the main problems is to
characterize their incentive compatibility in terms of an appropriate
payment rule.  Several such results were established in the
literature, starting with the one in \cite{GL77a} concerning Groves
mechanisms, originally proposed in \cite{Gro73}.  One of the earliest
characterization results was given in \cite{Mye81}, who considered
single object auctions in an imperfect information setting. In
\cite{Mil04} such characterizations are called `Myerson’s Lemma’. This
terminology was adopted in \cite{Rou16}, Chapter 3 of which, titled
`Myerson’s Lemma’, is concerned with a characterization of incentive
compatible single-parameter mechanisms which were introduced and
studied in \cite{AT01}.

As we explain below, both in this article and in Roughgarden’s book
such a characterization result is actually not proved. Most (but not
all) of the claims are rigorously established in \cite{Nis07}
in the context of randomized single-parameter mechanisms.

Given that this purported characterization of incentive compatible
single-parameter mechanisms is frequently referred to in the
literature (see e.g. \cite{HK07} or in \cite{Bab16}) we find it
justified to review these claims. We will provide an elementary proof
of the characterization result for two classes of allocation functions
considered in \cite{Rou16} and subsequently provide a proof of the
original claim of \cite{AT01} by appealing to more advanced results
from the theory of real functions.  We conclude by comparing our proof
to the one given in \cite{Kri02} and \cite{Bor15}.

\section{Preliminaries}
\label{sec:prelim}

We follow here the terminology of \cite{Rou16} that is slightly different
than the one originally used in \cite{AT01}.  In particular
\cite{Rou16} refers to a single-parameter mechanism and an allocation
rule, while \cite{AT01} refer to a one-parameter mechanism and load.

Each \bfe{single-parameter mechanism}
concerns sale of some `stuff’ to bidders and assumes
\begin{itemize}
\item a set of agents $\{1, \LL, n\}$, 

\item for every agent $i$, a value
  $v_i \ge 0$ which specifies $i$’s \bfe{private valuation}
  ``per unit of stuff’’ that he or she acquires.
\end{itemize}

In the auction the agents submit simultaneously their \bfe{bids},
which are their reported valuations ``per unit of stuff’’.  The
auctioneer receives the bids and determines how much `stuff’ each
agent receives and against which price.  So in contrast to the
single-item auctions each agent $i$ receives a possibly fractional
amount $a_i \ge 0$ of an object (here `a stuff’) he or she is
interested in.

An \bfe{allocation} is a vector $\mathbf{a} = (a_1, \dots, a_n)$,
where each $a_i \ge 0$ specifies the amount allocated to agent $i$.
A \bfe{payment} is a vector $\mathbf{p} = (p_1, \LL, p_ n)$, where
each $p_i \geq 0$ specifies the amount agent $i$ has to pay.

Each single parameter mechanism consists of an \bfe{allocation rule} 
\[
  \mathbf{a}: \mathbb{R}_{\ge 0}^n \to \mathbb{R}_{\ge 0}^n
\]
and a \bfe{payment rule}
\[
  \mathbf{p}: \mathbb{R}_{\ge 0}^n \to \mathbb{R}_{\ge 0}^n.
\]
Given a
vector of bids 
$\mathbf{b} = (b_1, \dots, b_n)$ such a mechanism
selects an allocation
$\mathbf{a}(\mathbf{b}) = (a_1(\mathbf{b}), \dots, a_n(\mathbf{b}))$
and a vector of payments
$\mathbf{p}(\mathbf{b}) = (p_1(\mathbf{b}), \dots, p_n(\mathbf{b}))$.

We assume that the \bfe{utility} of agent $i$ is defined by
\[
  u_i(\mathbf{b}) = v_i a_i(\mathbf{b}) - p_i(\mathbf{b}).
\]

We then say that a single-parameter mechanism is \bfe{incentive
  compatible} if for each agent $i$ truthful bidding, i.e., bidding
$v_i$, yields the best outcome regardless of bids of the other agents.
More formally, it means that for all agents $i$
\[
  u_i(v_i, \mathbf{b}_{-i}) \geq  u_i(b_i, \mathbf{b}_{-i}),
\]
for all bids $b_i$ of agent $i$ and all vectors of bids
$\mathbf{b}_{-i}$ of other agents, or equivalently---ignoring the
parameters $\mathbf{b}_{-i}$---that for all $y \geq 0$
\[
v_i a_i(v_i) - p_i(v_i)  \geq v_i a_i(y) - p_i(y).
\]

\section{A characterization result}

We say that a function $f: \mathbb{R}_{\ge 0} \to \mathbb{R}_{\ge 0}$ 
is \bfe{monotonically non-decreasing}, in short \bfe{monotone}, if 
\[
  0 \leq x \leq y \rightarrow f(x) \leq f(y).
\]

We say that an allocation rule $\mathbf{a}$ is \bfe{monotone} if for
every agent $i$ and every vector of bids $\mathbf{b}_{-i}$ of other
agents the function $a_i(\cdot, \mathbf{b}_{-i})$ is monotone.

The following result is stated in \cite{AT01}, \cite{Nis07}, and
\cite{Rou16}. (In \cite{Nis07} it is formulated as a result about
randomized single-parameter mechanisms but the proofs are the same for
the deterministic mechanisms considered here.)

\begin{theorem} \label{thm:char-ic}
\mbox{} \\[-2mm]
  \begin{enumerate}[(i)]

  \item If a mechanism $(\mathbf{a}, \mathbf{p})$ is incentive compatible
    then the allocation rule $\mathbf{a}$ is monotone.

  \item If the allocation rule $\mathbf{a}$ is monotone then for some
    payment rule $\mathbf{p}$ the mechanism $(\mathbf{a}, \mathbf{p})$ is
    incentive compatible.

  \item If the allocation rule $\mathbf{a}$ is monotone, then all
    payment rules $\mathbf{p}$ for which the mechanism
    $(\mathbf{a}, \mathbf{p})$ is incentive compatible differ by a
    constant.
  \end{enumerate}
\end{theorem}

There are some technically irrelevant differences between these three
references.  In \cite{AT01} instead of allocations loads are
considered, with the consequence that the loads are monotonically
non-increasing, though the authors also state that the results equally
apply to the setup that uses allocations.  In what follows, following
\cite{Nis07} and \cite{Rou16}, we use allocations. It leads to an
analysis of monotonically non-decreasing functions. Further, in the
last two references it is assumed that the payment rule yields 0
payment when bids are equal to 0, which makes it possible to drop in
$(iii)$ the qualification `up to a constant’. To make the discussion
applicable to arbitrary payment rules we do not adopt this
assumption.

Item $(i)$ is established in \cite{AT01} by appealing to the first
derivative, so under some assumptions about the load
function. However, a short argument given in \cite{Nis07} and
reproduced in \cite{Rou16} shows that no assumptions are needed.

In turn, item $(ii)$ is proved in \cite{AT01} `by picture’. A rigorous
proof is given in \cite{Nis07}, while in \cite{Rou16} only a `proof by
picture’ is provided for piecewise constant allocation rule and it is
mentioned that ``the same argument works more generally for monotone
allocation rules that are not piecewise constant’’.

Finally, in \cite{AT01} item $(iii)$ is claimed for arbitrary monotone
loads and allocation rules. But in the paper only a short proof sketch
is given that ends with a claim that ``To prove that all truthful
payment schemes take form (2), even when $\omega_i$ [the load rule] is
not smooth, we follow essentially the same reasoning as in the
[earlier given] calculus derivation.’’  However, this derivation
refers to load rules that are assumed to be smooth (actually only
twice differentiable, so that integration by parts can be applied),
while the characterization result is claimed for all monotone
allocation functions.

In \cite{Nis07} item $(iii)$ is established by reducing in the last
step the expression $\int_{0}^{x} z f’(z) dz$ to
$x f(x) - \int_{0}^{x} f(z) dz$.  We quote (adjusting the notation):
``[\dots] we have that $p(x) = \int_{0}^{x} z f’(z) dz$, and
integrating by parts completes the proof. (This seems to require the
differentiability of $f$, but as $f$ is monotone this holds almost
everywhere, which suffices since we immediately integrate.)’’
(Recall that a property holds \bfe{almost everywhere} if
it holds everywhere except at a set of measure 0, i.e., a set that can
be covered by a countable union of intervals the total length of which
is arbitrarily small.)
A minor point is that the initial part of the proof is incomplete as it
only deals with the right-hand derivative instead of the derivative.

Finally, in \cite{Rou16} about item $(iii)$ it is only stated without
proof ``We reiterate that these payments formulas [for the above two
classes of allocation functions] give \emph{the only possible} payment
rule that has a chance of extending the given allocation rule
\textbf{x} into a DSIC [i.e., incentive compatible] mechanism.’’  Also
here the formula (in the adjusted notation)
$p(x) = \int_{0}^{x} z f’(z) dz$ is derived by discussing only the
right-hand derivative.

In our view these arguments are incomplete as they do not take into
account some restrictions that need to be imposed on the use of
integrals and application of integration by parts. 
Note that, except in the final discussion, Riemann integration is assumed throughout.

\begin{remark} \label{rem:concerns}
\rm  

To start with,
integration by parts can fail for simple monotone functions, for
example those considered in \cite{Rou16}. Indeed, let for $q > 0$
\[
  H_q(x) :=\begin{cases}
    0 &\text{ if $0 \leq x \leq q$ }\\
    1 &\text{ if $x > q$ }
\end{cases}
\]
be an elementary step function with a single step at $x = q$. 

Take now $f = H_q$ with $q = 1/2$. Then $f’ = 0$ for $x \neq 1/2$ and
$f’$ is undefined for $x = 1/2$.  Consequently (defining $f’(1/2)$
arbitrarily)
\begin{equation}
  \label{equ:by-parts}
\int_{0}^{1} z f’(z) dz = 0 \neq 1/2 =  xf(x) \big{|}_0^1  - \int_{0}^{1} f(z) dz.
\end{equation}

Further, integration by parts can fail even if we insist on continuity.
Indeed, take for $f$ the Cantor function, see, e.g., \cite[pages
170-171]{Tao11}. It is monotone, continuous and almost everywhere
differentiable on $[0,1]$, with $f(0) = 0$, $f(1) = 1$ and $f’$ equal to 0
whenever defined.  Additionally, \eqref{equ:by-parts} holds for $f$,
as well.

Finally, there exists a monotone and everywhere differentiable
function $f$ for which the above integral $\int_{0}^{1} z f’(z) dz$
does not exist.  Indeed, as observed in \cite{Gof77}, there
exists a monotone and everywhere differentiable function
$f: [0,1] \rightarrow \mathbb{R}_{\ge 0}$ such that the integral
$\int_{0}^{1} f’(z) dz$ does not exist.  By a result of Lebesgue (see,
e.g., \cite{Bre08}) a bounded function defined on a bounded and closed
interval is Riemann integrable iff it is continuous almost
everywhere. But $f’$ is continuous almost everywhere on $[0,1]$ iff
the function $g(x) := xf’(x)$ is, so the claim follows.  \HB
\end{remark}

These points of concern motivate our subsequent considerations.  To
keep the paper self-contained we reprove items $(i)$ and $(ii)$, given
that the proofs are very short.

\section{An analysis}

Our analysis can be carried out without any reference
to mechanisms by reasoning about functions on
reals. We first rewrite the incentive compatibility condition as
\[
  p_i(y) - p_i(v_i) \geq v_i (a_i(y) - a_i(v_i)),
\]
which from now on
we analyze as the following condition on two functions
$f: \mathbb{R}_{\ge 0} \to \mathbb{R}_{\ge 0}$ and
$g: \mathbb{R}_{\ge 0} \to \mathbb{R}_{\ge 0}$:

\begin{equation}
  \label{equ:yax}
  \fa x, y: g(y) - g(x) \geq  x (f(y) -  f(x)).
\end{equation}

We are interested in solutions in $g$ given $f$.
We begin with the following obvious observation.

\begin{note} \label{not:implies}
The inequality \eqref{equ:yax} is equivalent to
\begin{equation}
  \label{equ:yax1}
\fa x, y: y (f(y) -  f(x)) \geq g(y) - g(x) \geq  x (f(y) -  f(x)).  
\end{equation}
\end{note}

\begin{proof}
  By interchanging in \eqref{equ:yax} $x$ and $y$ we get the
  additional inequality \\ $y (f(y) - f(x)) \geq g(y) - g(x)$.
\end{proof}

\begin{corollary} [\cite{Nis07,Rou16}] \label{cor:mon}
Suppose \eqref{equ:yax} holds. Then the function $f$ is monotone.  
\end{corollary}

\begin{proof}
  Assume $0 \leq x < y$.  By Note \ref{not:implies} \eqref{equ:yax1}
  holds. The inequalities in \eqref{equ:yax1} imply
  $(y-x) (f(y) - f(x)) \geq 0$, so $f(x) \leq f(y)$.
\end{proof}

This establishes item $(i)$ of Theorem \ref{thm:char-ic}.  To
investigate items $(ii)$ and $(iii)$ we study existence and uniqueness
of solutions of \eqref{equ:yax} in $g$.  The following
result establishes item $(ii)$. The proof is from \cite{Nis07}.

\begin{lemma} \label{lem:existence}
Suppose $f$ is monotone.
Then \eqref{equ:yax} holds for 

    \begin{equation}
      \label{equ:rp}
g(x) =  C + x f(x) - \int_{0}^{x} f(z) dz,
    \end{equation}
where $C$ is some constant.
\end{lemma}
Because $f$ is monotone $g$ is well defined (see, e.g., \cite{Rud76}).

\II

\begin{proof}
  By plugging the definition of $g$ in \eqref{equ:yax} we get after
  some simplifications

\begin{equation}
  \label{equ:int}
\int_{0}^{x} f(z) dz - \int_{0}^{y} f(z) dz \geq (x-y) f(y),
\end{equation}
which needs to be proved.
Two cases arise.
\II

\NI
\textbf{Case 1.} $x \geq y$.

Then
\[
\int_{0}^{x} f(z) dz - \int_{0}^{y} f(z) dz = \int_{y}^{x} f(z) dz \geq (x-y) f(y),
\]
where the last step follows by bounding the integral from below, since
by the monotonicity of $f$, we have
$f(y) \leq f(z)$ for $z \in [y, x]$. 
\II

\NI
\textbf{Case 2.} $y > x$.

Then 
\[
\int_{0}^{x} f(z) dz - \int_{0}^{y} f(z) dz = - \int_{x}^{y} f(z) dz \geq (x-y) f(y),
\]
where the last step follows by bounding the integral from above, since
by the monotonicity of $f$, we have
$f(z) \leq f(y)$ for $z \in [x, y]$.

So \eqref{equ:int} holds, which concludes the proof. 
\end{proof}

To deal with uniqueness let us first consider the case for which
the argument given in \cite{Nis07} can be justified.

\begin{lemma} \label{lem:everywhere}
  Suppose $f$ is everywhere differentiable. Then any two solutions $g$ of \eqref{equ:yax} 
  differ by a constant.
\end{lemma}  

\begin{proof}
  Suppose that \eqref{equ:yax} holds. By Note \ref{not:implies}
  \eqref{equ:yax1} holds.  Given an arbitrary $x \geq 0$ we first use
  it with $y = x+h$, where $h > 0$. Dividing by $h$ we then obtain
\[
  \frac{(x+h) (f(x+h) -  f(x))}{h} \geq \frac{g(x+h) - g(x)}{h} \geq  \frac{x (f(x+h) -  f(x))}{h}.
\]

By the assumption about $f$ 
\[
\lim_{h \to 0^+} \frac{(x+h) (f(x+h) -  f(x))}{h} =  \lim_{h \to 0^+} \frac{x (f(x+h) -  f(x))}{h} = x f'(x),
\]
so 
\begin{equation}
    \label{equ:one}
\lim_{h \to 0^+} \frac{g(x+h) - g(x)}{h} = x f'(x).
\end{equation}

Next, we use \eqref{equ:yax1} with $x = y+h$, where $h < 0$. Dividing by $h$ we then obtain
\[
  \frac{y (f(y) -  f(y+h))}{h} \leq \frac{g(y) - g(y+h)}{h} \leq  \frac{(y+h) (f(y) -  f(y+h))}{h},
\]
so replacing $y$ by $x$ and multiplying by $-1$ we get
\[
  \frac{x (f(x+h) -  f(x))}{h} \geq \frac{g(x+h) - g(x)}{h} \geq  \frac{(x+h) (f(x+h) -  f(x))}{h}.
\]

By the assumption about $f$ and $x$
\[
\lim_{h \to 0^-} \frac{(x+h) (f(x+h) -  f(x))}{h} =  \lim_{h \to 0^-} \frac{x (f(x+h) -  f(x))}{h} = x f'(x),
\]
so 
\begin{equation}
    \label{equ:two}
\lim_{h \to 0^-} \frac{g(x+h) - g(x)}{h} =  x f'(x).
\end{equation}

We conclude from \eqref{equ:one} and \eqref{equ:two} that $g'(x)$
exists and
\begin{equation}
  \label{equ:derivative}
g'(x) = x f'(x).
\end{equation}
Hence all solutions $g$ to
\eqref{equ:yax} have the same derivative and consequently differ by a
constant.
\end{proof}

\begin{remark} \label{rem:1}
\rm
  
The above proof coincides with the one given in \cite{Nis07}, except on
two points.  First, only \eqref{equ:one} is established there. This allows
one only to conclude that the right derivative of $g$ in $x$ exists;
to establish that $g’(x)$ exists also \eqref{equ:two} is
needed. More importantly, Nisan argued that all solutions $g$ to
\eqref{equ:yax} are of the form \eqref{equ:rp} given in Lemma
\ref{lem:existence}.  Under the assumption that $f$ is everywhere
differentiable this additional claim is a direct consequence of Lemmas
\ref{lem:existence} and \ref{lem:everywhere}.

Nisan’s argument for this point involved integration and integration by
parts. To justify it we need to assume that $f’$ is continuous.
Then by \eqref{equ:derivative} also $g’$ is continuous, which allows
us to use the Fundamental Theorem of Calculus. It yields that for some
constant $C$
\[
  g(x) =   C + \int_{0}^{x} g’(z) dz.
\]
Further, integration by parts of
$\int_{0}^{x} z f’(z) dz$ is then also justified since $f$ is
everywhere differentiable and $f’$ is integrable (see, e.g.,
\cite{Rud76}).  Then $\int_{0}^{x} z f’(z) dz$ exists and by
integration by parts
\[
  \int_{0}^{x} z f’(z) dz = x f(x) - \int_{0}^{x} f(z) dz,
\]
so \eqref{equ:derivative} and the last two equalities imply that $g$
is indeed of the form \eqref{equ:rp} given in Lemma
\ref{lem:existence}.
\HB
\end{remark}

In the remainder of this section we do not use integration or the existence of
solutions in the form \eqref{equ:rp}, but proceed directly from
\eqref{equ:yax}. This allows us to sidestep the associated
complications and show that the requirement of $f$ being everywhere
differentiable of Lemma \ref{lem:everywhere} can be substantially
weakened and, appealing to strong results from the theory of real
functions, can even be removed altogether.

For $x=0$ continuity (differentiability) means right continuity
(differentiability), which we will not mention or treat separately.

We first need an auxiliary result.

\begin{lemma} \label{lem:+-}
  Let $g_1$ and $g_2$ be two solutions of \eqref{equ:yax} and let
  $G = g_1 - g_2$.

  \begin{enumerate}  [(i)]
  \item $G$ is continuous.

  \item If $f$ is continuous at $x$, then $G$ is differentiable at $x$ and $G’(x) = 0$.
  \end{enumerate}

\end{lemma}

\begin{proof}

\NI  
$(i)$
  By Note \ref{not:implies} \eqref{equ:yax1} holds for $g_1$ and
  $g_2$. By using it with $y = x+h$ for $g_1$ and for $g_2$ we obtain
\begin{equation}
\label{equ:G}
0 \leq | G(x+h) - G(x) |  \leq h(f(x+h) - f(x)).
\end{equation}

We have $h(f(x+h) - f(x)) \leq hf(x+h)$ for $h > 0$ and
$\leq -hf(x)$ for $h < 0$.
But by Corollary \ref{cor:mon} $f$ is monotone, hence
$\lim_{h \to 0} | G(x+h) - G(x) | = 0$, which establishes the claim.
\II

\NI
$(ii)$
Take some $x \geq 0$.  By \eqref{equ:G} for $h \neq 0$
  \[
0 \leq \Big{|}\frac{G(x+h) - G(x) }{h} \Big{|}  \leq |f(x+h) - f(x)|,
  \]
which implies the claim.
\end{proof}

\NI
Note that the continuity of $G$ holds for any $f$.

The following result with an elementary proof covers in a unified way item $(iii)$ of Theorem
\ref{thm:char-ic} for two classes of allocation functions considered
in \cite{Rou16}, piecewise constant and differentiable ones.

A function $f: \mathbb{R}_{\ge 0} \to \mathbb{R}_{\ge 0}$ is called
\bfe{piecewise continuous} if it has at most a finite number of
discontinuities in every bounded interval. Thus discontinuities can
occur only at isolated points separated by open intervals of
continuity. Piecewise constant and step functions are special
cases. This definition is a straightforward generalization to
$\mathbb{R}_{\ge 0}$ of the usual one for functions with a bounded
domain.

\begin{theorem} \label{thm:1} Suppose $f$ is
  piecewise continuous. Then any two solutions $g$ of \eqref{equ:yax} differ by a constant.
\end{theorem}
\begin{proof}
  Let $g_1$ and $g_2$ be two solutions of \eqref{equ:yax} and let
  $G = g_1 - g_2$.

  Let $f$ be piecewise continuous with discontinuities
  $q_1 < q_2 < \ldots $ and consider the intervals $I_0 = [0 , q_1)$
  ($\emptyset$ if $q_1 = 0$), $I_i = (q_i, q_{i+1})$ ($i \geq 1$),
  with $I_N = (q_N, \infty)$ if $f$ has a finite number $N>0$ of
  discontinuities and $I_0 = [0, \infty) = \mathbb{R}_{\ge 0}$ if
  $N = 0$.

If $f$ has an infinite number of
  discontinuities, $\lim_{i \to \infty} q_i = \infty$ since there can
  be only finitely many of them in any bounded interval. Hence, the
  $I_i$ and $q_i$ together cover the whole of $\mathbb{R}_{\ge 0}$.

$f$ is continuous on each $I_i$, so by Lemma \ref{lem:+-}$(ii)$ $G$ is
constant on $I_i$, say $G = C_i$ on $I_i$. Since $G$ is continuous
everywhere by Lemma \ref{lem:+-}$(i)$,
$C_0 = G(q_1) = C_1 = G(q_2) = \ldots $, so for some constant $C$, $G = g_1 - g_2 = C$.
\end{proof}

Theorem \ref{thm:1} can be generalized to a wider class of functions
whose discontinuity sets may have limit points (accumulation points),
at least to some degree. We give a simple example of a monotone
function $f$, for which Theorem \ref{thm:1} does not apply but the
stronger result presented below does.

Let
\[ 
f(x) = \sum_{ n = 1 }^\infty { 2^{-n} } H_{ 1 - { 2^{-n} } } (x),
\]
where $H_q$ is defined in Remark \ref{rem:concerns}.
It is not piecewise continuous, but has an infinite set of
discontinuities $\{ 1/2 , 3/4 , \ldots \}$ with a single limit point
$1$. Note that $f$ happens to be continuous at $x = 1$, but this might also have been otherwise.

Functions like $f$ and more complicated ones having discontinuity sets
with limit points of limit points, etc. can, to some extent, be dealt
with by adapting the proof of Theorem \ref{thm:1} and appealing to the
well-known Bolzano-Weierstrass theorem (BW for short, see, e.g.,
\cite{Bre08}).

Given a set $S \subseteq \mathbb{R}_{\ge 0}$, we denote by $S^{(1)}$
the set of its limit points (which need not be in $S$) and define
$S^{(n+1)} = (S^{(n)})^{(1)}$ for $n \geq 1$. A set $S$ is called
\bfe{first species} of type $n-1$ if $S^{(n)} = \emptyset$ and
$S^{(m)} \neq \emptyset$ for $m <n$ \citep{Bre08}. Such a set has limit
points, limit points of limit points, etc., up to level $n-1$. A first
species set of type $0$ has no limit points.

\begin{theorem} \label{thm:limitpoints}
Suppose the discontinuity set of $f$ is first species of type $n \ge 0$.
Then any two solutions $g$ of \eqref{equ:yax} differ by a constant.
\end{theorem}
\begin{proof}
Let $g_1$ and $g_2$ be two solutions of \eqref{equ:yax} and let $G = g_1 - g_2$.
Let the discontinuity set of $f$ be $S$ and use induction with respect to the type of $S$.

If $n = 0$, $S^{(1)} = \emptyset$, so $f$ can have only finitely many discontinuities in every bounded interval. Otherwise, there would be a limit point in some bounded and closed interval by BW. Hence, $f$ is piecewise continuous and the case $n = 0$ corresponds to Theorem \ref{thm:1}.
  
Assume the theorem holds for all $n \leq k$ for some $k > 0$ and consider $S$ of type $k+1$.
The elements of $S^{(k+1)}$ are the limit points of level $k+1$. Recall that these need not be elements of $S$.
Since $S^{(k+2)} = \emptyset$, $S^{(k+1)}$ does not have limit points, so there can be only finitely many elements of $S^{(k+1)}$ in every bounded interval by BW as before.

Now let the elements of $S^{(k+1)}$ be $q_1 < q_2 < \ldots $ and consider the intervals $I_i$ ($ i \geq 0$)
  as in the proof of Theorem \ref{thm:1}. Together with the $q_i$,
  they cover the whole of $\mathbb{R}_{\ge 0}$ as in the previous proof.

However, the $S \cap I_i$ may still have $q_i$ and/or $q_{i+1}$ as limit points, which means 
the $S \cap I_i$ need not be of type $\leq k$ but can still be of type $k+1$, so the induction 
hypothesis cannot be applied to the $I_i$. Therefore, for fixed $i$ and sufficiently small $\delta > 0$ 
consider a non-empty bounded and closed subinterval
$J_i = J_i(\delta) = [ q_i + \delta, q_{i+1} - \delta]$ of $I_i$ (or, if the number of
limit points is a finite number $N$,
$J_N = J_N(\delta) = [ q_N + \delta , \infty)$).

Now $S \cap J_i$ can no longer have $q_i$ and/or $q_{i+1}$ as limit points.
Hence, it is of type $\leq k$ and the induction hypothesis applies to $J_i$, so
$G$ is constant on $J_i$, say $G = C_i$ on $J_i$. Since $G$ is continuous
everywhere by Lemma \ref{lem:+-}$(i)$ and
$\lim_{\delta \to 0} J_i = [ q_i, q_{i+1}]$ (or, if the number of limits points is $N$, $\lim_{\delta \to 0} J_N = [ q_N, \infty)$),
we get $G(q_i) = C_i = G(q_{i+1})$. The final step of the proof is the same as in the proof of Theorem \ref{thm:1}.
\end{proof}

Unfortunately, the above result does not cover all monotone functions
$f: \mathbb{R}_{\ge 0} \to \mathbb{R}_{\ge 0}$. Indeed, a monotone
function may be discontinuous on the set $\mathbb{Q}_{\ge 0}$ of
non-negative rational numbers, see, e.g., \cite{Rud76}, and
$\mathbb{Q}_{\ge 0}$ is not first species, since
$\mathbb{Q}_{\ge 0}^{(1)} = \mathbb{R}_{\ge 0}$ and
$\mathbb{R}_{\ge 0}^{(1)} = \mathbb{R}_{\ge 0}$.

This limitation can be circumvented by appealing to a strong result of
Goldowski and Tonelli. Recall first that a function
$f: \mathbb{R}_{\ge 0} \rightarrow \mathbb{R}_{\ge 0}$ is differentiable
\bfe{nearly everywhere} if it is differentiable except at a countable
number of points.
Note that \bfe{nearly everywhere} implies
\bfe{almost everywhere}. We need

\begin{theorem}[\cite{Gol28,Ton30,Sak37}] \label{thm:Gol}
Let $G: \mathbb{R}_{\ge 0} \to \mathbb{R}_{\ge 0}$ be a function such
    that

\begin{itemize}
\item $G$ is continuous,
  
\item $G$ is differentiable nearly everywhere,
  
\item $G’ \geq 0$ almost everywhere.
\end{itemize}

Then $G$ is monotone.
\end{theorem}
This leads directly from Lemma \ref{lem:+-} to the desired conclusion.

\begin{theorem} \label{thm:full} Any two solutions $g$ of
  \eqref{equ:yax} differ by a constant.
\end{theorem} 

\begin{proof}
  Let $g_1$ and $g_2$ be two solutions of \eqref{equ:yax} and let
  $G = g_1 - g_2$. By Lemma \ref{lem:+-}$(i)$ $G$ is continuous.
  
  A monotone function is continuous nearly everywhere (see, e.g.,
  \cite{Rud76}). So by Lemma \ref{lem:+-}$(ii)$ $G$ is differentiable nearly
  everywhere and $G’ = 0$ nearly everywhere. Hence by Theorem
  \ref{thm:Gol} both $G$ and $-G$ are monotone, i.e., $G$ is constant.
\end{proof}

The above theorem justifies item $(iii)$ of Theorem \ref{thm:char-ic}.
The following result summarizes the results of this section.

\begin{theorem} \label{thm:1a}

Inequality  \eqref{equ:yax}
  holds iff $f$ is monotone and for some constant $C$
\[
g(x) = C + x f(x) - \int_{0}^{x} f(z) \: dz.
\]
\end{theorem}
\begin{proof}
  By Corollary \ref{cor:mon}, 
  Lemma \ref{lem:existence}, and Theorem \ref{thm:full}.
\end{proof}

\section{Discussion} \label{sec:Discussion}

Results closely corresponding to our uniqueness result (Theorem
\ref{thm:full}) were also presented in  \cite{Kri02} (and
its second edition \cite{Kri09}) and \cite{Bor15}. (The customary name
of these results is Revenue Equivalence.)  Krishna considers in
Chapter 5 a setup with a seller that has one indivisible object to
sell and $n$ potential buyers, while B\"{o}rgers considers in Chapter
2 a setup in which there is just one potential buyer.  In \cite{Kri09}
the equivalent of our function $f$ is defined as an integral
representing the probability that a buyer gets the object, while in
\cite{Bor15} $f$ corresponds to the probability of selling the object
to the buyer. However, a close inspection of the proofs of these
Revenue Equivalence results reveals that they do not depend on the
actual form of $f$.

Further, ignoring the differences in the setup, the corresponding
proofs in both books are from the mathematical point of view
essentially the same. As the arguments in the latter one are more
detailed, we discuss them here, but using our notation.

The proof of the corresponding result (Proposition 2.2) in
\cite{Bor15} is not based on the equivalent of our function $f$ but
instead deals, in Lemma 2.2, with the function $u$ (representing
utility) defined by
\[
  u(x) := xf(x) - g(x),
\]
and states that for all $x$ for which $u$ is differentiable,
\[
  u’(x) = f(x).
\]

Lemma 2.2 also establishes that the function $u$ is monotone and convex.
Then in Lemma 2.3 it is shown that
\[
  u(x) = u(0) +  \int_{0}^{x} f(z) dz,
\]
which is equivalent to \eqref{equ:rp} by taking $C = u(0)$, so the
uniqueness result (Lemma 2.4 (Revenue Equivalence)) corresponding to
our Theorem \ref{thm:full}, follows.

Lemma 2.3 is a direct consequence of two results from 
\cite{RF10}, namely, that convexity implies
absolute continuity (a notion we leave undefined here) and that every
absolutely continuous function is equal to the integral of its
derivative.

Note, however, that the latter result (Theorem 10 of \cite{RF10}) is
the Fundamental Theorem of Calculus (FTC) for the Lebesgue integral, a
fact not mentioned in \cite{Kri02} and in \cite{Bor15} deducible only
indirectly from footnote 2 in Chapter 2. So the proofs of the
uniqueness result (the Revenue Equivalence) presented in \cite{Kri02}
and \cite{Bor15} \emph{crucially} rely on the Lebesgue theory of
integration.  In contrast, our proof is much more elementary: it does
not rely on \emph{any} form of integration and appeals only to the
notion of derivative.  Only the existence result (Lemma
\ref{lem:existence}) relies on the Riemann integral.  Having said
this, apart from its complications, use of the Lebesgue integral
yields a very efficient proof.

Both Krishna and B\"{o}rgers establish in their books appropriate
Revenue Equivalence results for other mechanisms.  In particular,
B\"{o}rgers considers in Chapters 3 and 4 of \cite{Bor15} Bayesian
mechanisms and dominant mechanisms, each time for $n$ buyers.  In both
cases he establishes the corresponding Revenue Equivalence result
(Lemma 3.4 and Proposition 4.2) by explaining that the reasoning
provided in Chapter 2 can be repeated.

Given that in both setups the crucial inequalities are the
counterparts of \eqref{equ:yax} (considered separately for each
buyer), it follows that both results can be alternatively proved using
our Theorem \ref{thm:full}. We conclude that our more elementary
approach can be applied to other mechanisms than the single-parameter
mechanism considered in Section \ref{sec:prelim}.

Finally, we show that a generalization of Theorem \ref{thm:1a} allows
one to provide alternative, more elementary proofs of Revenue
Equivalence for two other types of auctions, considered in
\cite{Kri09} in Chapters 14 and 16. These auctions are concerned with
multiple objects, leading to functions $f$ and $g$ having to deal with
vectors.

We use the following notation.  For functions
$f: \mathbb{R}^n_{\ge 0} \to \mathbb{R}^n_{\ge 0}$ and
$g: \mathbb{R}^n_{\ge 0} \to \mathbb{R}_{\ge 0}$ and
$\x \in \mathbb{R}^n_{\ge 0}$ we introduce the functions
$f_{\x}: \mathbb{R}_{\ge 0} \to \mathbb{R}_{\ge 0}$ and
$g_{\x}: \mathbb{R}_{\ge 0} \to \mathbb{R}_{\ge 0}$ defined by
\[
  f_{\x}(t) =  f(t \x) \cdot \x,
\]
where  $\cdot$  is the inner product (dot product), and
\[
  g_{\x}(t) = g(t \x).
\]

We then say that $f: \mathbb{R}^n_{\ge 0} \to \mathbb{R}^n_{\ge 0}$ is
\bfe{monotone} if for each $\x \in \mathbb{R}^n_{\ge 0}$ the function
$f_{\x}$ is monotone.

We now establish the following result which generalizes Theorem
\ref{thm:1a} to dimension $n > 1$.  (Note that using the substitution
$u(z) = z x$ we have \\
$\int _0^x f(u) \: du = \int_{0}^{1} f(z x)x \: dz$.)

\begin{theorem} \label{thm:2}
\rm   
For two functions
$f: \mathbb{R}^n_{\ge 0} \to \mathbb{R}^n_{\ge 0}$ and
$g: \mathbb{R}^n_{\ge 0} \to \mathbb{R}_{\ge 0}$ the inequality 
\begin{equation}
  \label{equ:1}
  \fa \x, \y: g(\y) - g(\x) \geq  (f(\y) -  f(\x)) \cdot \x
\end{equation}
holds iff $f$ is monotone and for some constant $C$
\begin{equation}
  \label{equ:1a}
g(\x) = C + f(\x) \cdot \x - \int_{0}^{1} f(z\x) \cdot \x \: dz.
\end{equation}
\end{theorem}

\begin{proof}
First note that \eqref{equ:1} holds iff 
\begin{equation}
  \fa \x, x,y: g_{\x}(y) - g_{\x}(x)  \geq x (f_{\x}(y) - f_{\x}(x)),
  \label{equ:2}
\end{equation}
since for each $\x \in \mathbb{R}^n_{\ge 0}$
\[
  \fa x,y: g_{\x}(y) - g_{\x}(x) = g(y \x) - g(x \x)
\]
and by linearity of the inner product
\[
\fa x,y: x (f_{\x}(y) - f_{\x}(x)) = (f(y \x) -  f(x \x)) \cdot x \x.
\]

By Theorem \ref{thm:1}
\eqref{equ:2} holds iff for all $\x$ the function $f_{\x}$ is monotone and
for some constant $C_{\x}$
\begin{equation}
  g_{\x}(x) =  C_{\x} + x f_{\x}(x) - \int_{0}^{x} f_{\x}(z) \: dz.
  \label{equ:3}
\end{equation}

But $C_{\x} = g_{\x}(0) = g(\mathbf{0})$, so the constant $C_{\x}$
does not depend on $\x$.  Further, $g(\x) = g_{\x}(1)$, and
$1 f_{\x}(1) = f(\x) \cdot \x$, so by putting $x=1$ we see
that\eqref{equ:3} implies \eqref{equ:1a}.

But also \eqref{equ:1a} implies \eqref{equ:3}, which can be seen by
using \eqref{equ:1a} with $x \x$ instead of $\x$.

\end{proof}
\smallskip

Let us return now to \cite{Kri09}.  In Chapter 14 he
studies multiunit auctions in which multiple identical objects are
available. The relevant inequality (14.1) on page 204, capturing the
expected payment in an equilibrium for a player, corresponds to
\eqref{equ:1}. Theorem \ref{thm:2} then provides an alternative proof
of his Proposition 14.1 stating that
\begin{quote}
``The equilibrium payoff (and payment) functions of any bidder
in any two multiunit auctions that have the same allocation rule differ at most
by an additive constant.''
\end{quote}

Krishna's proof relies (implicitly) on the Lebesgue integral. We
adopted from his proof the idea of reasoning about the functions
$f_{\x}$ and $g_{\x}$.  In Chapter 16 of his book he studies auctions
in which one can bid for a set of nonidentical objects. (They are
called in the computer science literature combinatorial auctions.)
Krishna explains that ``The proof is \emph{identical} to that of
Proposition 14.1.''  (italic used by the author). Consequently, our
approach also yields an alternative proof of Revenue Equivalence for
combinatorial auctions.

\renewcommand{\bibfont}{\normalfont\small}
\bibliographystyle{apacite}
\bibliography{e}

\end{document}